\documentclass{llncs}

\usepackage{amsmath, amssymb}
\usepackage{graphicx}
\usepackage{url}
\usepackage{stmaryrd}
\usepackage{bussproofs}

\usepackage{proof}

\newcommand{\ignore}[1]{}
\newcommand{\ital}[1]{{\em #1}}

\newcommand{\ra}{\rightarrow}

\newcommand{\lprolog}{$\lambda$Prolog}
\newcommand{\hhf}{$hohh$}

\newcommand{\dtlc}{dependently typed $\lambda$-calculus}

\newcommand{\ruvappt}{{\scriptsize APP$_\text{t}$}}
\newcommand{\ruvpit}{{\scriptsize PI$_\text{t}$}}
\newcommand{\ruvinito}{{\scriptsize INIT$_\text{o}$}}
\newcommand{\ruvappo}{{\scriptsize APP$_\text{o}$}}
\newcommand{\ruvabso}{{\scriptsize ABS$_\text{o}$}}

\renewcommand{\vec}[1]{\overrightarrow{#1}}

\newcommand{\subst}[1]{[#1]}

\newcommand{\iprove}[2]{\sequent{#1}{#2}}
\newcommand{\lfprove}[2]{#1 \, \vdash \, #2}
\newcommand{\hastype}[2]{hastype \ #1 \ #2}

\newcommand{\encTerm}[1]{\langle #1 \rangle}

\newcommand{\encExtP}[2]{\llbracket #1 \rrbracket^{+}_{#2}}
\newcommand{\encExtN}[1]{\llbracket #1 \rrbracket^{-}}
\newcommand{\termUV}[2]{#1 \sqsubset_o #2}
\newcommand{\formulaUV}[2]{#1 \sqsubset_t #2}

\newcommand{\sequent}[2]{#1 \longrightarrow #2}

\newcommand{\typedlambda}[3]{\lambda #1 \mbox{:} #2 . #3}

\newcommand{\typedpi}[3]{\Pi #1 \mbox{:} #2 . #3}
\newcommand{\typedpis}[3]{\Pi \vec{#1 \mbox{:} #2} . #3}
\newcommand{\app}{\ }
\newcommand{\emptyctx}{\cdot}
\newcommand{\oftype}[2]{#1 : #2}
\newcommand{\ctx}{\mbox{\sl ctx}}
\newcommand{\kind}{\mbox{\sl kind}}
\newcommand{\type}{\mbox{\sl Type}}

\newcommand{\nullctx}{\mbox{\sl null-ctx}}
\newcommand{\kindctx}{\mbox{\sl kind-ctx}}
\newcommand{\typectx}{\mbox{\sl type-ctx}}
\newcommand{\typekind}{\mbox{\sl type-kind}}
\newcommand{\pikind}{\mbox{\sl pi-kind}}
\newcommand{\varfam}{\mbox{\sl var-fam}}
\newcommand{\varobj}{\mbox{\sl var-obj}}
\newcommand{\pifam}{\mbox{\sl pi-fam}}
\newcommand{\absfam}{\mbox{\sl abs-fam}}
\newcommand{\appfam}{\mbox{\sl app-fam}}
\newcommand{\appobj}{\mbox{\sl app-obj}}
\newcommand{\absobj}{\mbox{\sl abs-obj}}

\newcommand{\ie}{{\it i.e.}}

\newtheorem{definition'}{Definition}
\newtheorem{remark'}{Remark}
\newtheorem{theorem'}{Theorem}
\newtheorem{corollary'}{Corollary}
\newtheorem{lemma'}{Lemma}
\newtheorem{proposition'}{Proposition}

\begin{document}

\long\def\comment#1{}

\title{Redundancies in Dependently Typed Lambda Calculi
 and Their Relevance to Proof Search}

\author{Zachary Snow, David Baelde and Gopalan Nadathur}
\institute{
Department of Computer Science and Engineering                    \\
University of Minnesota                                           \\
4-192 EE/CS Building, 200 Union Street SE, Minneapolis, MN 55455  \\
snow@cs.umn.edu, david.baelde@gmail.com, gopalan@cs.umn.edu
}

\maketitle

\begin{abstract}
Dependently typed $\lambda$-calculi such as the Logical Framework (LF)
are capable of representing relationships between terms through
types. By exploiting the ``formulas-as-types'' notion, such calculi
can also encode the correspondence between formulas and their proofs
in typing judgments. As such, these calculi provide a natural yet
powerful means for specifying varied formal systems. Such
specifications can be transformed into a more direct form that uses
predicate formulas over simply typed $\lambda$-terms and that thereby
provides the basis for their animation using conventional logic
programming techniques. However, a naive use of this idea is fraught
with inefficiencies arising from the fact that dependently typed
expressions typically contain much redundant typing information. We
investigate syntactic criteria for recognizing and, hence, eliminating
such redundancies. In particular, we identify a property of bound
variables in LF types called {\it rigidity} and formally show that
checking that instantiations of such variables adhere to typing
restrictions is unnecessary for the purpose of ensuring that the
overall expression is well-formed. We show how to exploit this
property in a translation based approach to executing specifications
in the Twelf language. Recognizing redundancy is also relevant to
devising compact representations of dependently typed expressions. We
highlight this aspect of our work and discuss its connection with
other approaches proposed in this context.

\end{abstract}

\section{Introduction}
\label{sec:introduction}

There is a significant, and growing interest in mechanisms for
specifying, prototyping and reasoning about formal systems that are
described by syntax-directed rules.  Dependently typed
$\lambda$-calculi such as the Logical Framework (LF)
\cite{harper93jacm} provide many conveniences from a specification
perspective in this context: such calculi facilitate the use of a
higher-order approach to describing the syntax of formal objects, they
allow relationships between terms to be captured in an elegant way
through type dependencies, and they allow proof-checking to be
realized through type-checking. Such calculi can also be given a
logic programming interpretation by exploiting the isomorphism between
formulas and types \cite{howard80}. The Twelf 
system \cite{pfenning99cade}, based on this idea, 
has been used successfully in specifying, prototyping, and reasoning about
varied formal systems.

While a system like LF provides interesting and useful ways to factor
typing properties of terms and relationships between terms, this 
ability is not essential to its specification applications. In
particular, it is possible also to use predicate based 
descriptions over simply typed $\lambda$-terms to similar effect. In
fact, it is possible to formally present a translation of dependently
typed specifications into a predicate logic form that uses the
properties of being a type and of being a term of a certain type
\cite{felty89phd,felty90cade}. Such a translation lends itself to the
possibility of using an implementation of a conventional logic
programming language like \lprolog~
\cite{nadathur88iclp,teyjus.website} to animate specifications written
in LF \cite{snow10ppdp}. Moreover, if the translation preserves the structure of the
original specification, it would be possible to view the dependently
typed descriptions as meta-programs and to reason about them using
techniques for reasoning about the generated predicate based specifications.

Unfortunately, the reality with respect to a straightforward
translation does not quite fit this promise. The essential source of
the problem is the fact that expressions in a dependently typed
language typically contain much redundant type information. When such
information is preserved in a translation, it leads to a predicate
logic program that is not efficient to execute. The resulting
extraneous typing constraints destroy also the transparency of
the correspondence, and thereby interfering with the meta-program view
and the reasoning possibilities. 

These observations provide the motivation for the work we describe
here: identifying redundancies in LF expressions. More specifically,
we want to elucidate syntactic criteria for determining such
unnecessary information that can, for instance, be exploited in a
translation based approach to implementing LF specifications. We
describe a property of bound variables in LF types called {\it
  rigidity} and show in a formal sense that knowledge of the specific
instantiations of such variables is unnecessary from the perspective
of checking if the expression is well-formed. While our observations
are driven by a particular application, they also have a relevance in
other contexts such as that of devising compact representation of
proofs \cite{reed08tcs}. We discuss these connections in the paper. 

In the next section we describe the \dtlc\ and LF.  We then introduce
a translation from LF to a predicate logic that preserves provability, and
describe how redundancies in LF derivations can impact the performance of the
generated logic program.  In Section~\ref{sec:redundancy} we develop a technique
for identifying and eliminating such redundancies.  Then in Section~\ref{sec:optimization}
we show how it can be applied to improve the translation, and identify some
important practical extensions to the translation.  We conclude the paper
with a discussion of possible future directions for this work.  

\section{The Edinburgh Logical Framework}
\label{sec:lf}

The Edinburgh Logical Framework~\cite{harper93jacm} (LF)
is a dependently typed $\lambda$-calculus with three categories of 
expressions:
{\em kinds}, {\em types} or {\em type families} that are
classified by kinds and {\em objects} or {\em terms} that are
classified by types. We assume two denumerable sets of 
variables, one for objects and the other for types. We use $x$
and $y$ to denote object variables, $u$ and $v$ to denote type
variables and $w$ to denote either. Letting $K$ range over kinds, $A$
and $B$ over types, and $M$ and $N$ over object terms, the syntax of
LF expressions is given by the following rules:
\begin{center}
\begin{tabular}{ccl}
    $K$ & $:=$ & $\type\ |\ \typedpi{x}{A}{K}$ \\ 
    $A$ & $:=$ &
        $u\ |\ \typedpi{x}{A}{B}\ |\ \typedlambda{x}{A}{B}\ |\ A~M $ \\ 
    $M$ & $:=$& $x\ |\ \typedlambda{x}{A}{M}\ |\ M~N $ \\
\end{tabular}
\end{center}
Expressions of any of these kinds will be denoted by $P$ and $Q$. Here,
$\Pi$ and $\lambda$ are operators that associate a type with a 
variable and bind its free occurrences over the expression
after the period.  We write $P[N_1/x_1,\ldots,N_n/x_n]$ to denote a simultaneous
substitution with renaming to avoid variable capture. We write $A
\rightarrow P$ for $\typedpi{x}{A}{P}$ when $x$ does not appear free
in $P$, and abbreviate $\typedpi{x_1}{A_1}{\ldots\typedpi{x_n}{A_n}{P}}$
by $\typedpis{x}{A}{P}$.

\begin{figure}
\begin{center}
   \AxiomC{}
   \RightLabel{\nullctx}
   \UnaryInfC{$\lfprove{}{\emptyctx\ \ctx}$}
   \DisplayProof

\medskip

   \AxiomC{$\lfprove{\Gamma}{K\ \kind} \quad
     \lfprove{}{\Gamma\ \ctx}\quad u \notin dom(\Gamma)$} 
   \RightLabel{\kindctx}
   \UnaryInfC{$\lfprove{}{\Gamma, u : K\ \ctx}$}
   \DisplayProof

\medskip

   \AxiomC{$\lfprove{\Gamma}{\oftype{A}{\type}} \quad
     \lfprove{}{\Gamma\ \ctx}\quad x \notin dom(\Gamma)$}
   \RightLabel{\typectx}
   \UnaryInfC{$\lfprove{}{\Gamma, x : A\ \ctx}$}
   \DisplayProof

\medskip

\begin{tabular}{cc}
   \AxiomC{$\lfprove{}{\Gamma\ \ctx}$}
   \RightLabel{\typekind}
   \UnaryInfC{$\lfprove{\Gamma}{\type\ \kind}$}
   \DisplayProof

&
   \AxiomC{$\lfprove{\Gamma}{\oftype{A}{\type}} \quad
     \lfprove{\Gamma,\oftype{x}{A}}{K\ \kind}$}
   \RightLabel{\pikind}
   \UnaryInfC{$\lfprove{\Gamma}{\typedpi{x}{A}{K}\ \kind}$}
   \DisplayProof
\end{tabular}

\medskip
\begin{tabular}{cc}
   \AxiomC{$\lfprove{}{\Gamma\ \ctx} \quad \oftype{u}{K} \in \Gamma$}
   \RightLabel{\varfam}
   \UnaryInfC{$\lfprove{\Gamma}{\oftype{u}{K^\beta}}$}
   \DisplayProof

&
   \AxiomC{$\lfprove{}{\Gamma\ \ctx} \quad \oftype{x}{A} \in \Gamma$}
   \RightLabel{\varobj}
   \UnaryInfC{$\lfprove{\Gamma}{\oftype{x}{A^\beta}}$}
   \DisplayProof
\end{tabular}

\medskip

   \AxiomC{$\lfprove{\Gamma}{\oftype{A}{\type}} \quad \lfprove{\Gamma,
       \oftype{x}{A}}{\oftype{B}{\type}}$}
   \RightLabel{\pifam}
   \UnaryInfC{$\lfprove{\Gamma}{\oftype{(\typedpi{x}{A}{B})}{\type}}$}
   \DisplayProof

\medskip
\begin{tabular}{cc}
   \AxiomC{$\lfprove{\Gamma}{\oftype{A}{\type}} \quad
     \lfprove{\Gamma,\oftype{x}{A}}{\oftype{B}{K}}$}
   \RightLabel{\absfam}
   \UnaryInfC{$\lfprove{\Gamma}{\oftype{(\typedlambda{x}{A}{B})}{(\typedpi{x}{A^\beta}{K})}}$}
   \DisplayProof

&

   \AxiomC{$\lfprove{\Gamma}{\oftype{A}{\typedpi{x}{B}{K}}} \quad
     \lfprove{\Gamma}{\oftype{M}{B}}$}
   \RightLabel{\appfam}
   \UnaryInfC{$\lfprove{\Gamma}{\oftype{(A\app M)}{(K[M/x])^\beta}}$}
   \DisplayProof
\end{tabular}

\medskip
\begin{tabular}{cc}
   \AxiomC{$\lfprove{\Gamma}{\oftype{A}{\type}} \quad
     \lfprove{\Gamma,\oftype{x}{A}}{\oftype{M}{B}}$}
   \RightLabel{\absobj}
   \UnaryInfC{$\lfprove{\Gamma}{\oftype{(\typedlambda{x}{A}{M})}{(\typedpi{x}{A^\beta}{B})}}$}
   \DisplayProof

&

   \AxiomC{$\lfprove{\Gamma}{\oftype{M}{\typedpi{x}{A}{B}}} \quad
     \lfprove{\Gamma}{\oftype{N}{A}}$}
   \RightLabel{\appobj}
   \UnaryInfC{$\lfprove{\Gamma}{\oftype{(M\app N)}{(B[N/x])^\beta}}$}
   \DisplayProof
\end{tabular}
\end{center}
  \caption{Rules for Inferring LF Assertions}
  \label{fig:lf-rules}
\end{figure}

The type correctness of LF expressions is assessed relative to
contexts that are finite collections of assignments of types and kinds
to variables (we use $\cdot$ to denote the empty context).
LF deals with assertions of the following four forms:
\begin{center}
$\lfprove{}{\Gamma\ \ctx} \qquad \lfprove{\Gamma}{K \ \kind} \qquad
\lfprove{\Gamma}{\oftype{A}{K}} \qquad \lfprove{\Gamma}{\oftype{M}{A}}$ 
\end{center}
The first assertion signifies that $\Gamma$ is a
well-formed context. The remaining assertions mean
respectively that, relative to a (well-formed) context $\Gamma$, $K$
is a well-formed kind, $A$ is a well-formed type of kind $K$ and $M$
is a well-formed object of type $A$.
Figure~\ref{fig:lf-rules} presents the rules for deriving such
assertions. The inference rules
allow for the derivation of an assertion of the form
$\lfprove{\Gamma}{\oftype{M}{A}}$ only when $A$ is in normal form. To
verify such an assertion when $A$ is not in normal form, we first
derive $\lfprove{\Gamma}{\oftype{A}{\type}}$ and then verify
$\lfprove{\Gamma}{\oftype{M}{A^\beta}}$. A similar observation applies
to $\lfprove{\Gamma}{\oftype{A}{K}}$. 

Well-typed LF expressions admit a $\beta\eta$-long form.
Types of $\beta\eta$-long form $(u \ M_1 \ \ldots \ M_n)$
are called \emph{base types}.
In the following, we shall only consider LF derivations whose end
assertion only contains expressions in $\beta\eta$-long form.
Notice that every expression in the entire derivation must then
also be in $\beta$-normal form. This in turn means that in judgments of
the forms $\oftype{(\typedlambda{x}{A}{B})}{(\typedpi{x}{A'}{K})}$ and 
$\oftype{(\typedlambda{x}{A}{M})}{(\typedpi{x}{A'}{B})}$ it must be
the case that $A$ and $A'$ are identical, and that normalization need
not be considered in the use of the \varfam\ and \varobj\ rules.
Additionally, we shall modify inference rules so that all expressions
in the entire derivation remain in $\beta\eta$-long form
| the resulting system is referred to as \emph{canonical} LF.
For objects, \appobj\ and \varobj\ are replaced by the following
big-step application rule (which we refer to as \emph{backchaining}
due to its logic programming interpretation):
\[ \infer{\lfprove{\Gamma}{\oftype{M}{A\subst{N_1/x_1\ldots N_n/x_n}}}}{
      (\oftype{y}{\typedpis{x}{B}{A}})\in\Gamma &
      \lfprove{\Gamma}{\oftype{N_1}{B_1}} & \ldots &
      \lfprove{\Gamma}{\oftype{N_n}{B_n\subst{N_1/x_1\ldots N_{n-1}/x_{n-1}}}}
   }
\]
The rules for type families are changed in the same way.

\bigskip

The logic programming interpretation of LF is based on viewing 
types as formulas. More specifically, a specification or program in
this setting is given by a context. This starting context, also called
a {\em signature}, essentially describes the vocabulary for
constructing types and asserts the existence of particular inhabitants
for some of these types. Against this backdrop, questions can be asked
about the existence of inhabitants for certain other types. Formally,
this amounts to asking if an assertion of the form
$\lfprove{\Gamma}{\oftype{M}{A}}$ has a derivation. However, the
object $M$ is left unspecified | it is to be extracted
from a successful derivation. Thus, the search for a derivation of 
the assertion is driven by the structure of $A$ and the types
available from the context.  

A concrete illustration of the paradigm is useful for later
discussions.
\ignore{\footnote{The example of appending lists has been chosen
  here for its conciseness and because it allows for an easy connection
  with more traditional forms of logic programming. The primary
  application domain of Twelf is in specifying (and reasoning about)
  formal systems such as evaluators and interpreters for languages,
  type assignment calculi and proof systems.
  We do not intend to present those aspects here,
  as well as many other features of Twelf.}
}
Consider a signature or  context $\Gamma$ comprising the 
following assignments in sequence: 
\begin{tabbing}
\qquad\=\qquad\quad\=\qquad\=\kill
\>$\oftype{nat}{\type}$, \quad $\oftype{z}{nat}$, \quad $\oftype{s}{nat \rightarrow nat}$, \\
\>$\oftype{list}{\type}$, \quad $\oftype{nil}{list}$, \quad $\oftype{cons}{nat \rightarrow list \rightarrow list}$, \\ 
\>$\oftype{append}{list \rightarrow list \rightarrow list \rightarrow \type}$, \\
\>$\oftype{appNil}{\typedpi{K}{list}{append~nil~K~K}}$, \\
\>$\oftype{appCons}{\typedpi{X}{nat}{\typedpi{L}{list}{\typedpi{K}{list}{\typedpi{M}{list}{}}}}}$\\
\>\>$(append~L~K~M) \rightarrow$\\
\>\>\>$(append~(cons~X~L)~K~(cons~X~M))$
\end{tabbing}
We can ask if there is some term $M$ such that the judgment
\begin{center}
$\lfprove{\Gamma}{\oftype{M}{append~(cons~z~nil)~(cons~(s~z)~nil)~(cons~z~(cons~(s~z)~nil))}}$
\end{center}
is derivable.

Furthermore, as Twelf allows for instantiatable meta-variables in the type $A$,
we can ask \ital{which} list results from appending $(cons~z~nil)$ and $(cons~(s~z)~nil)$;
in the following, $L$ is such a variable: \begin{center}
$\lfprove{\Gamma}{\oftype{M}{append~(cons~z~nil)~(cons~(s~z)~nil)~L}}$.
\end{center}
Here Twelf responds by instantiating $L$ to
$(cons~z~(cons~(s~z)~nil))$ and instantiating $M$ with an LF object (proof 
term) of type 
$append~(cons~z~nil)~(cons~(s~z)~nil)~(cons~z~(cons~(s~z)~nil))$.
Sometimes the logic program $\Gamma$ does not constrain a meta-variable, and
so it is left uninstantiated in the proof term.
Here, the interpretation is that
the proof term is valid for \emph{any} instantiation of the meta-variable by a 
term of the right type.

\section{A translation to predicate logic}
\label{sec:translation}

We now consider translating LF specifications into the logic of
higher-order hereditary Harrop formulas, also known as \hhf\ logic
\cite{miller91apal}. Intuitively, this logic is similar to Horn clause
logic, except that it allows predicates to have simply typed
$\lambda$-terms as arguments, it permits quantification over
(non-predicate) function variables and it includes universal
quantifiers and embedded implications in goals and the bodies of
clauses. Althought \hhf\ does not permit dependent types, it has been
shown that these dependencies can be systematically encoded by using
predicates over the simpler form of $\lambda$-terms.
This encoding is particularly interesting because it provides a way to
utilize efficient implementations of \hhf\ logic like the Teyjus
system \cite{teyjus.website} in realizing an LF-based logic
programming language. 

The general idea of the encoding, first proposed by 
Felty~\cite{felty89phd,felty90cade}, is to first encode LF objects
and types in a way that erases type dependencies, and to recover
those relationships in the encoding of LF type judgments.
One can then prove that the encoding is sound and complete by
providing a mappings between LF derivation and \hhf\ derivations
of the encoded judgments.
However, it is important to note that in the perspective of
using the encoding for logic programming, a stronger correctness
result is needed.
Instead of considering only closed LF assertions,
\ie\ checking a given type judgment,
we are interesting in finding objects inhabiting a particular type,
\ie\ searching for a derivation of an LF assertion
with a meta-variable for the object.
Therefore, our correctness result should also state that any instantiation
of that meta-variable is actually an LF encoding.

\begin{figure}
\begin{tabbing}
\qquad\qquad\qquad\=\qquad\=\kill
$hastype~z~nat$ \\
$\forall n.~ hastype~n~nat \supset hastype~(s~n)~nat$ \\
$hastype~nil~list$ \\
$\forall n.~
       hastype~n~nat \supset \forall l.~ \hastype~l~list \supset
       hastype~(cons~n~l)~list$ \\
$\forall l.~ hastype~l~list \supset
      hastype~(appNil~l)~(append~nil~l~l)$ \\
$\forall x.~ hastype~x~nat \supset
   \forall l.~ hastype~l~list \supset
   \forall k.~ hastype~k~list \supset$ \\
\>$\forall m.~ hastype~m~list \supset
     \forall a.~ hastype~a~(append~l~k~m) \supset$ \\
\>$hastype~(appCons~x~l~k~m~a)~(append~(cons~x~l)~k~(cons~x~m))$
\end{tabbing}
\caption{Simple translation of the LF specification for $append$}
\label{fig:simplified-append-translation}
\end{figure}

We shall only give an intuition and example of our translation,
referring the reader to \cite{snow10ms} for details and proofs.
Our translation proceeds in the same general fashion as Felty's: LF objects
and types are first encoded as \hhf\ terms.  Next the $hastype$ predicate
relates \hhf\ terms representing LF objects with \hhf\ terms representing
the LF types of those objects.  For instance, given an LF object $z$ of
type $nat$, we relate \hhf\ encodings $z'$ and $nat'$ thus: $hastype~z'~nat'$.
As an example,
the Twelf specification of $append$ translates into the clauses in
Figure~\ref{fig:simplified-append-translation}. From these clauses, we
can, for example, derive the goal $hastype~(cons~(s~z)~nil)~list$
and we could search for terms $X$ satisfying the following goal:
$$hastype~X~(append~(cons~z~nil)~(cons~(s~z)~nil) (cons~z~(cons~(s~z)~nil)))$$
Unfortunately, this program does not correspond exactly to the usual
$append$ logic program in \hhf. Specifically, whenever a goal
$hastype~p~(append~l~k~m)$ is proved, each list $l$, $k$, and $m$ is
``type-checked'' by deriving a proof of, for example, the subgoal
$hastype~l~list$.
This involves a recursion over the entire structure of the
list, and thereby introduces a quadratic complexity to the
fundamentally linear operation
of appending lists.

As we shall see, a meta-theoretical analysis of LF derivations
can be used to justify the removal of some of those typing constraints.
This study of derivations is best carried out directly in LF,
leading to interesting results in their own right, some of which
may be applicable beyond our translation problem.

\section{Redundancy in LF derivations}
\label{sec:redundancy}
The redundancy evoked above can be viewed from the LF standpoint alone.
Consider a variable $y$ of type $\typedpis{x}{B}{A}$ which might be
used to derive some judgment
$\lfprove{\Gamma}{\oftype{M}{A\subst{N_1/x_1\ldots N_n/x_n}}}$:
\[ \infer{\lfprove{\Gamma}{\oftype{M}{A\subst{N_1/x_1\ldots N_n/x_n}}}}{
      (\oftype{y}{\typedpis{x}{B}{A}})\in\Gamma &
      \lfprove{\Gamma}{\oftype{N_1}{B_1}} & \ldots &
      \lfprove{\Gamma}{\oftype{N_n}{B_n\subst{N_1/x_1\ldots N_{n-1}/x_{n-1}}}}
   }
\]
It is reasonable to assume that when checking that an object
has a particular type, or when searching for objects of a particular type,
the type has been checked to be valid first,
\ie\ that $\lfprove{\Gamma}{\oftype{A}{\type}}$ has a derivation.
It is often the case that some of the typing judgments 
$\lfprove{\Gamma}{\oftype{N_i}{B_i}}$ are superfluous in the sense
that they can be found almost verbatim in the derivation that
$A$ is a type.  Furthermore, it is possible to detect statically many of those cases,
just by examining the occurrences of $x_i$ in $A$.
The idea is that if $x_i$ occurs in $A$ in such a way that
$N_i$ will be found in $A\subst{N_1/x_1\ldots N_n/x_n}$, whatever the
other $N_j$ are, then the premise $\oftype{x_i}{N_i}$ can be
safely omitted.

Formally, we use the notion of a \ital{rigid occurrence} that is expressed
by the judgment $\termUV{\vec{x};\cdot; x_i}{B}$
defined in Figure~\ref{fig:ruvs} to
characterize some of these cases.

\begin{figure}
\[
   \AxiomC{$y_i$ distinct elements of $\delta$}
   \RightLabel{\ruvinito}
   \UnaryInfC{$\termUV{\Gamma; \delta; x}{x~\vec{y}}$}
 \DisplayProof
\quad\quad
   \AxiomC{$y \notin \Gamma$ and $\termUV{\Gamma; \delta; x}{M_i}$ for some $i$}
   \RightLabel{\ruvappo}
   \UnaryInfC{$\termUV{\Gamma; \delta; x}{y\ \vec{M}}$}
 \DisplayProof
\]
\[
   \AxiomC{$\termUV{\Gamma; \delta, y; x}{M}$}
   \RightLabel{\ruvabso}
   \UnaryInfC{$\termUV{\Gamma; \delta; x}{\typedlambda{y}{A}{M}}$}
 \DisplayProof
\]
  \caption{Rigidly occurring variables in LF objects}
  \label{fig:ruvs}
\end{figure}

\begin{theorem}
\label{theorem:rigid-variables}
  Let $\vec{N}$ be a vector of LF objects,
  $\vec{x}$ a vector of variables, and $\vec{B}$ of canonical LF types,
  all of same length.
  Let $\Gamma$ and $\Delta$ be LF contexts, $\delta$ be $dom(\Delta)$.
  Let $\Gamma_0 = x_1 : B_1, \ldots, x_n : B_n$.
  Let $\typedpis{x}{B}{A}$ be a canonical type, where $A$ is a base type.
  Suppose that there are derivation of: \begin{itemize}
  \item $\termUV{\vec{x};\delta; x_i}{M}$
  \item $\lfprove{\Gamma,\Gamma_0,\Delta}{
                     \oftype{M}{A}}$
  \item $\lfprove{\Gamma,\Delta\subst{\vec{N/x}}}{
                     \oftype{M\subst{\vec{N/x}}}{A\subst{\vec{N/x}}}}$
  \end{itemize}
  Then there is a derivation of
  $\lfprove{\Gamma}
           {\oftype{N_i}{B_i \subst{N_1 / x_1, \ldots, N_{i-1} / x_{i-1}}}}$.
\end{theorem}

This theorem establishes a sort of substitution inversion:
having an abstract and an instantiated derivation,
we show that one can recover the derivation that was substituted,
that is
  $\lfprove{\Gamma}
           {\oftype{N_i}{B_i \subst{N_1 / x_1, \ldots, N_{i-1} / x_{i-1}}}}$.
Given the nature of that statement, it is not surprising that we find
in \ruvinito\ a condition reminiscent of higher-order patterns,
a fragment of higher-order unification where most general unifiers are
guaranteed, thanks to the ability to invert substitutions.

\begin{proof}[Theorem~\ref{theorem:rigid-variables}]
We proceed by induction on the rigidity derivation.
Walking simultaneously through the two LF derivations, following the
path given by the rigidity derivation, we eventually reach a point where we have
on the one hand a derivation of
$\oftype{x_i \ \vec{y}}{T \subst{y_1/z_1 \ldots y_k/z_k}}$
with $B_i = \typedpis{z}{C}{T}$,
and on the other a derivation of
$\oftype{N_i \ \vec{y}}{T \subst{y_1/z_1 \ldots y_k/z_k} \subst{N_1/x_1\ldots N_{i-1}/x_{i-1}}}$.
The bound variables $\vec{y}$ being distinct, the substitution 
$\subst{y_1/z_1\ldots y_k/z_k}$ is simply a renaming and can be inverted.
We obtain a derivation of
$\oftype{N_i \ \vec{z}}{T \subst{N_1/x_1\ldots N_{i-1}/x_{i-1}}}$ and
finally $\oftype{N_i}{B_i \subst{N_1/x_1\ldots N_{i-1}/x_{i-1}}}$.
\qed
\end{proof}

\begin{remark}
Note that it would be unsound to allow in \ruvinito\
any application $x \ \vec{N}$
rather than $x \ \vec{y}$ for distinct bound variables $\vec{y}$.
With such a rule the rigidity lemma the above theorem is no longer true.
For example, in a signature with
$\oftype{num}{nat \ra \type}$ and
$\oftype{num_n}{\typedpi{n}{nat}{(num \app n)}}$,
we obtain a counter-example with $M = \lambda x.~ x \ z$ and
$N = t$:
we have $\lfprove{\Gamma}{\oftype{(t \app z)}{(num \ z)}}$ and
\[\lfprove{\Gamma, \oftype{x}{(nat \ra num\ z)}}{\oftype{(x \ z)}{(num \ z)}}\]
but not $\lfprove{\Gamma}{\oftype{t}{nat \rightarrow num \app z}}$.
\end{remark}

\subsection{Application to proof search}
There are several ways to exploit this property about LF derivations,
and not just in the context of a translation, but in the more general
setting of proof search.
We come back to the problem of eliminating redundancies in the
rule corresponding to backchaining on some element of the LF context:
\[ \infer{\lfprove{\Gamma}{\oftype{M}{A\subst{N_1/x_1\ldots N_n/x_n}}}}{
      (\oftype{y}{\typedpis{x}{B}{A}})\in\Gamma &
      \lfprove{\Gamma}{\oftype{N_1}{B_1}} & \ldots &
      \lfprove{\Gamma}{\oftype{N_n}{B_n\subst{N_1/x_1\ldots N_{n-1}/x_{n-1}}}}
   }
\]
Eliminating redundancies here corresponds to limiting the number of redundant
subderivations investigated during search.

We first consider recognizing rigid occurrences of some variables $x_i$
in the target type $A$. We formalize this as $\formulaUV{\vec{x};x_i}{A}$,
defined by the following rules:
\[
    \AxiomC{$\termUV{\Gamma; \cdot; x}{M_i}$ for some $M_i$}
    \RightLabel{\ruvappt}
    \UnaryInfC{$\formulaUV{\Gamma; x}{c \vec{M}}$}
\DisplayProof
\quad
\quad
  \AxiomC{$\formulaUV{\Gamma, y; x}{B}$}
  \RightLabel{\ruvpit}
  \UnaryInfC{$\formulaUV{\Gamma; x}{\typedpi{y}{A}{B}}$}
\DisplayProof
\]

\begin{theorem} \label{theorem:inheritance}
  Let $\vec{N}$ be a vector of LF objects,
  $\vec{x}$ a vector of variables, and $\vec{B}$ of canonical LF types,
  all of same length.
  Let $\Gamma$ and $\Delta$ be LF contexts, $\delta$ be $dom(\Delta)$.
  Let $\Gamma_0 = x_1 : B_1, \ldots, x_n : B_n$.
  Let $\typedpis{x}{B}{A}$ be a canonical type, where $A$ is a base type.
  Suppose that there are derivation of: \begin{itemize}
  \item $\termUV{\vec{x}; x_i}{A}$
  \item $\lfprove{\Gamma,\Gamma_0,\Delta}{
                     \oftype{A}{\type}}$
  \item $\lfprove{\Gamma,\Delta\subst{\vec{N/x}}}{
                     \oftype{A\subst{\vec{N/x}}}{\type}}$
  \end{itemize}
  Then there is a derivation of
  $\lfprove{\Gamma}
           {\oftype{N_i}{B_i \subst{N_1 / x_1, \ldots, N_{i-1} / x_{i-1}}}}$.
\end{theorem}

\begin{proof}
Similarly to Theorem~\ref{theorem:rigid-variables},
we walk through the type structure, following the path given by rigidity.
Eventually, we reach \ruvappt\ and invoke directly the previous theorem.
\end{proof}

From a practical viewpoint, this theorem allows us to statically analyze
an LF specification (which constitutes the initial LF context) and
discard some premises of the backchaining rules derived from that
specification, without losing soundness. This is currently done in our translation.

There are yet more redundancies in this same style.
We have used a rigid occurrence of some variable
$x_i$ in $A$ to retrieve a typing derivation for $N_i$ from
the derivation that $A \subst{N_1 / x_1, \ldots, N_{i-1} / x_{i-1}}$
is a valid type,
but we might also \ital{extend} the application of rigidity to retrieve
some information from the typing derivation for some $N_j$.
Given that we already have a derivation of
$\lfprove{\Gamma}{
     \oftype{\typedpi{x_1}{B_1}{\ldots\typedpi{x_n}{B_n}{A}}}{\type}}$,
we clearly have a derivation of
$\lfprove{\Gamma, \oftype{x_1}{B_1}, \ldots, \oftype{x_{j-1}}{B_{j-1}}}{
             \oftype{B_j}{\type}}$.
We also have a derivation of
$\lfprove{\Gamma}{
       \oftype{N_j}{B_j \subst{N_1 / x_1, \ldots, N_{j-1} / x_{j-1}}}}$,
either directly as one of the premises when $x_j$ is not rigid in $A$
or through Theorem~\ref{theorem:inheritance} when the corresponding premise
has been elided.
From this derivation we can also conclude that
$\lfprove{\Gamma}{
   \oftype{B_j \subst{N_1 / x_1, \ldots, N_{j-1} / x_{j-1}}}{\type}}$
has a derivation. We can hence finally apply Theorem~\ref{theorem:inheritance} to
these derivations to conclude that we do indeed have a derivation of 
$\lfprove{\Gamma}{
   \oftype{N_i}{B_i \subst{N_1 / x_1, \ldots, N_{i-1} / x_{i-1}}}}$.

\subsection{Related work}
Reed~\cite{reed08tcs} approaches the problem of eliminating redundancies in LF
from a different perspective, and with a different goal: that of reducing the \ital{size} of proof-terms
yielded during logic programming search,
motivated by the fact that in some applications
proof-terms must be transmitted or manipulated.  He does so by developing a
technique for identifying redundancies in terms, through a notion of \ital{strictness}
that is similar to rigidity, that he uses to identify
sub-terms of LF objects that can be reconstructed, either from the
types of nearby sub-terms, or from the type of an object itself.  
He describes two modes for omitting sub-terms, \ital{synthesis} based omission
and \ital{inheritance} based omission, and uses strictness to determine which
kind of omission, if any, is possible.
In omission by inheritance, knowledge of a term's type is used to elide
(and later reconstruct) type derivations for sub-terms.
For example, if $x+y$ is known to have type $nat$, then we automatically
known that $x$ has type $nat$, given that $+$ has type $nat \ra\ nat \ra\ nat$.
This is similar to what we described in Theorem~\ref{theorem:inheritance}.
In omission by synthesis, the types of nearby sub-terms are used to elide and
eventually reconstruct a given sub-term, when the sub-term being omitted appears
(in a sufficient manner) in said type.
For example, if $x=y$ is well-typed and $x$ has type $A$ we can deduce
that $y$ has type $A$ as well. This is similar to the additional application of
rigidity that we have described.

The main difference with Reed's work lies in the motivation.
Reeds work focuses on optimizing an LF
object (that is, a proof term) for size by eliminating redundant parts of the
object itself, and without particular concern for how such a term is discovered.
We are concerned with optimizing search, and we use the redundancy analysis
to avoid searching for parts of the typing derivation, but we still produce
a complete LF proof term.

\section{Optimizing the Twelf translation}
\label{sec:optimization}

We have presented a technique for identifying redundancies in LF derivations
and identified a few ways to use it in the context of proof-search.
Carrying these observations to the context of our translation to
\hhf---where we are also concerned with ensuring that all \hhf\ objects
discovered as instantiations of meta-variables actually correspond to
encodings of LF objects---is not entirely trivial.

\subsection{Meta-variables in objects}

Building on Theorem~\ref{theorem:inheritance}, we have developed in \cite{snow10ppdp}
an optimized translation from LF specifications to \hhf\ logic. The
part of this translation that removes redundant typing judgments is
based on the mapping on types presented in
Figure~\ref{fig:optimization}. The translation of context items of the
form $\oftype{x}{A}$ in an LF specification is given by
$(\encExtP{A}{\langle\rangle}\app x)$, where $\langle \rangle$ denotes
an empty sequence of variables; this operation is lifted to LF specifications by
distribution to each item in the specification.  The translation of a type $A$ for
which an inhabitant $M$ is sought is correspondingly given by
$(\encExtN{A}\app \langle M \rangle)$. 
Notice that these translations are guided solely by the type $A$; this is highlighted by the fact that
the translation actually returns a formula \ital{abstracted} over the proof-term.
This translation is illustrated by its application to the example
Twelf specification considered in Section~\ref{sec:lf} that yields the
clauses shown in Figure~\ref{fig:optimized-append}, which
should be contrasted with the ones in
Figure~\ref{fig:simplified-append-translation}.

\begin{figure*}
  \centering
  \begin{align*}
    \encExtP{\typedpi{x}{A}{B}}{\Gamma} :=&\
      \begin{cases}
        \lambda M.~ \forall x.~ \top \supset \encExtP{B}{\Gamma, x}(M \app x)
          & \text{if}\ \formulaUV{\Gamma; x}{B} \\
        \lambda M.~ \forall x.~
            \encExtN{A}(x) \supset \encExtP{B}{\Gamma, x}(M \app x)
          & \text{otherwise}
      \end{cases} \\
    \encExtP{N}{\Gamma} :=&\ \lambda M.~ hastype~M \app
    \encTerm{N}\qquad\text{if}\ N\ \text{is a base type}
    \\[8pt]
    \encExtN{\typedpi{x}{A}{B}} :=&\
      \lambda M.~ \forall x.~
         \encExtP{A}{\cdot}(x) \supset \encExtN{B}(M \app x) \\
    \encExtN{N} :=&\ \lambda M.~hastype\app M \app \encTerm{N}
    \qquad\text{if}\ N\ \text{is a base type}
  \end{align*}
  \caption{Optimized translation of LF specifications and judgments to \hhf}
  \label{fig:optimization}
\end{figure*}

\begin{figure*}
\begin{tabbing}
\qquad\qquad\qquad\=\qquad\=\kill
$hastype~z~nat$, $\forall n.~ hastype~n~nat \supset hastype~(s~n)~nat$, \\
$hastype~nil~list$, $\forall n.~ hastype~n~nat \supset 
   \forall l.~ hastype~l~list \supset hastype~(cons~n~l)~list$, \\
$\forall l.~ \top \supset
      hastype~(appNil~l)~(append~nil~l~l)$, \\
$\forall x.~ \top \supset
   \forall l.~ \top \supset
   \forall k.~ \top \supset
   \forall m.~ \top \supset
   \forall a.~ hastype~a~(append~l~k~m) \supset$ \\
\>$hastype~(appCons~x~l~k~m~a)~(append~(cons~x~l)~k~(cons~x~m))$
\end{tabbing}
\caption{Optimized translation of the LF specification for $append$}
\label{fig:optimized-append}
\end{figure*}

We have proved the optimized translation correct. The statement of its
correctness is slightly complicated by the fact that it requires that all 
\hhf\ terms correspond to LF expressions, so that we can use the translation to
generate actual LF proof terms.

\begin{theorem}[Optimized translation correctness]
  \label{theorem:optimization-correctness}
  Let $\Gamma$ be an LF specification such that
  $\lfprove{}{\Gamma\ \ctx}$ has a
  derivation, $A$ an LF type such that
  $\lfprove{\Gamma}{\oftype{A}{\type}}$ has a derivation.
  Then, for any LF object $M$ such that
  $\lfprove{\Gamma}{\oftype{M}{A}}$ has a derivation,
  $\iprove{\encExtP{\Gamma}{}}{\encExtN{\oftype{M}{A}}}$ is derivable.
  Moreover, if
  $\iprove{\encExtP{\Gamma}{}}{\encExtN{A}(M)}$ for an arbitrary \hhf\ term $M$,
  then it must be that $M = \encTerm{M'}$ for some canonical LF object
  such that $\lfprove{\Gamma}{\oftype{M'}{A}}$ has a derivation.
\end{theorem}

The proof of the following relies on Theorem~\ref{theorem:inheritance}
to recover typing judgments that have been optimized away.
In addition, it shows that \hhf\ terms must be well-formed LF objects.
Note that this theorem implies that proof-search for encoded
LF typing judgments will always fully instantiate the meta-variable
corresponding to the object
| otherwise, a dummy instantiation of that variable
  would still yield a valid derivation invalidating our theorem.

Unfortunately, we have not been able to exploit the extended
redundancy analysis to further optimize our translation; it has proven
difficult to ensure that \hhf\  
meta-variables are instantiated by LF encodings while still maintaining
an efficient translation.  This is due to the fact that, in eliminating
redundancies in this fashion, we must eventually obtain a typing derivation in
a setting without these optimizations, which could reduce or even destroy the
effectiveness of such eliminations.

\subsection{Meta-variables in types}

Note that, while we have proved Theorem~\ref{theorem:optimization-correctness}
for closed LF types, we have not yet considered the meaning
of meta-variables in such types, what it means when a meta-variable is not
bound during search, nor whether bindings for them are
correct.  Here there are two approaches.

Recall the interpretation of remaining meta-variables after proof-search,
in both \lprolog\ and Twelf: the goal actually holds
for \emph{any} term $t$ of the right type.
In particular, upon successful \lprolog\ search for an encoded LF query,
remaining meta-variables in the type can be instantiated by any encoding of an
LF object.
This can be done after the main proof search, by searching for an inhabitant of
the corresponding type.
We can then extend this treatment even to meta-variables that \ital{are} bound
during search, by simply checking \ital{after} search succeeds that
the meta-variables have been properly instantiated.
Once we have checked that the initial type has been instantiated into
a closed valid type in that way,
Theorem~\ref{theorem:optimization-correctness} applies.
In practice, this process is less intensive
than proof search proper, and tends not to be overly expensive.

Going further, it should in fact be the case that, under our translations,
no meta-variable could possibly be bound to the encoding of
an LF term of incorrect type or to something that is not even an encoding.
\ignore{In the example above, \lstinline!B! was bound
to the term \lstinline!void!, which clearly has the correct LF type $tm~unit$.}
The intuition here is that the only time a meta-variable is bound in the
logic programs generated by the translation is when it is matched with the head
of a clause.  Since the original specification is valid,
any such matching clause
should impose only the correct type on the meta-variable.
However, the statement and proof of this theorem is not at all obvious, and is
further stymied by the fact that it isn't clear how exactly this extension to
Twelf, which we are seeking to emulate, should behave.

\section{Conclusion and Future Work}
\label{sec:conclusion}

We have considered in this paper a translation from specifications in
the dependently typed $\lambda$-calculus LF to a predicate logic
over simply typed $\lambda$-terms. This translation is motivated by a
desire to utilize implementations of proof search in the latter logic
to realize LF-based proof search. A key task in making such a
translation effective is that of identifying and, subsequently,
eliminating redundancies in LF expressions and
derivations. Specifically, we have described a property of bound
variables in types that makes it unnecessary to type-check their
instantiations in ensuring that expressions that use such types are
well-formed. We note that our proof of such redundancy is based 
directly on the properties of LF expressions 
and derivations. Thus, our observation is of larger interest than just
the translation task at hand. 

The work described here can be extended in at least two ways. First,
it should be possible to enhance our techniques for identifying
redundancies. We have presented one such extension already through a
more inclusive definition of the rigidity property. 
However LF derivations contain significant redundancies and we
believe it is possible to carry out a richer analysis towards identifying
these based on syntactic properties. 
Second, we can think of applying the specific techniques developed for
detecting such redundancies to contexts different from translation.
We have already discussed the relationship between our work and that
of Reed. An understanding of the differences between our
system and his could eventually lead to a better, and provably
correct, ability to shorten LF proof terms that are needed in
applications such as that of proof-carrying code \cite{necula97popl}.
Moreover the usefulness of these ideas need not be limited to 
translation and compact representation of LF expressions: any
application of LF that requires type-checking, such as automatic
meta-theorem proving, could benefit from methods for discovering
repetitive type information.

\section{Acknowledgements}
This work has been supported by the NSF grant CCF-0917140. Opinions,
findings, and conclusions or recommendations  expressed in this paper
are those of the authors and do not necessarily reflect the views of
the National Science Foundation.

\bibliographystyle{alpha}
\newcommand{\etalchar}[1]{$^{#1}$}

\end{document}